\documentclass[11pt]{patmorin}
\usepackage[utf8]{inputenc}
\usepackage{amsmath,amsfonts,amssymb,amsthm,graphicx,graphics}
\usepackage{graphicx}

\usepackage{hyperref}
\usepackage[dvipsnames]{xcolor}
\definecolor{linkblue}{named}{Blue}
\hypersetup{colorlinks=true, linkcolor=linkblue,  anchorcolor=linkblue,
citecolor=linkblue, filecolor=linkblue, menucolor=linkblue,
urlcolor=linkblue} 

\newcommand{\etal}{\emph{et al.}}

\newtheorem{theorem}{Theorem}[section]
\newtheorem{corollary}[theorem]{Corollary}
\newtheorem{lemma}[theorem]{Lemma}

\DeclareMathOperator{\rank}{rank}

\newcommand{\R}{\mathbb{R}}

\title{\MakeUppercase{Compatible Connectivity-Augmentation \newline of Planar Disconnected Graphs}}

\author{Greg Aloupis,\thanks{Department of Computer Science, Tufts University, 
                             \email{aloupis.greg@gmail.com}}\,\,
       Luis Barba,\thanks{School of Computer Science, Carleton University
                          and Département d'Informatique, 
                          Université Libre de Bruxelles,
                          \email{lbarbafl@ulb.ac.be}}\,\,
       Paz Carmi,\thanks{Department of Computer Science,
                         Ben-Gurion University of the Negev,
                         \email{carmip@cs.bgu.ac.il}}\,\,
       Vida Dujmović,\thanks{School of Computer Science 
                             and Electrical Engineering,
                             University of Ottawa,
                             \email{vida.dujmovic@uottawa.ca}}\,\,
       Fabrizio Frati,\thanks{School of Information Technologies,
                              The University of Sydney,
                              \email{fabrizio.frati@sydney.edu.au}}\,\,
       and Pat Morin\thanks{School of Computer Science, Carleton University,
                            \email{morin@scs.carleton.ca}}}

\begin{document}

\begin{titlepage}

\maketitle
\begin{abstract}
Motivated by applications to graph morphing, we consider the following
\emph{compatible connectivity-augmentation problem}: We are given
a labelled $n$-vertex planar graph, $\mathcal{G}$, that has $r\ge 2$
connected components, and $k\ge 2$ isomorphic planar straight-line drawings,
$G_1,\ldots,G_k$, of $\mathcal{G}$. We wish to augment $\mathcal G$
by adding  vertices and edges to make it connected in such a way that
these vertices and edges can be added to $G_1,\ldots,G_k$ as points and
straight-line segments, respectively, to obtain $k$ planar straight-line
drawings isomorphic to the augmentation of $\mathcal G$.  We show
that adding $\Theta(nr^{1-1/k})$ edges and vertices to $\mathcal{G}$
is always sufficient and sometimes necessary to achieve this goal.
The upper bound holds for all $r\in\{2,\ldots,n\}$ and $k\ge 2$ and is
achievable by an algorithm whose running time is $O(nr^{1-1/k})$ for
$k=O(1)$ and whose running time is $O(kn^2)$ for general values of $k$.
The lower bound holds for all $r\in\{2,\ldots,n/4\}$ and $k\ge 2$.
\end{abstract}

\end{titlepage}


\section{Introduction}

Consider the following problem, which will be more carefully formalized
below.  We are given several different labelled planar straight-line drawings
(or simply drawings) of the same disconnected labelled graph, $\mathcal G$.
We wish to make $\mathcal G$ connected by adding vertices and edges in
such a way that these vertices and edges can also be added to the 
drawings of $\mathcal G$ while preserving planarity.  
The objective is to do this while minimizing the number
of edges and vertices added.  As the example in Figure~\ref{fig:bad-example} shows, it is not always possible to just add edges to $\mathcal G$; sometimes additional vertices are necessary.

\begin{figure}[hb]
  \centering{
    \begin{tabular}{ccc}
      \includegraphics{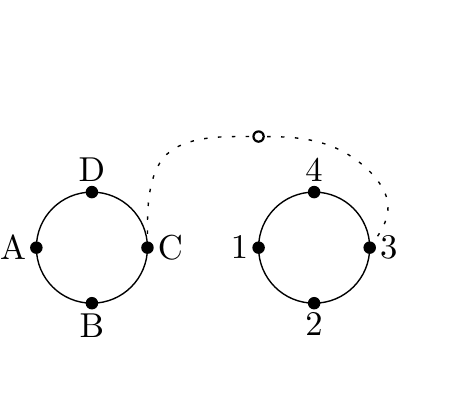} &
      \includegraphics{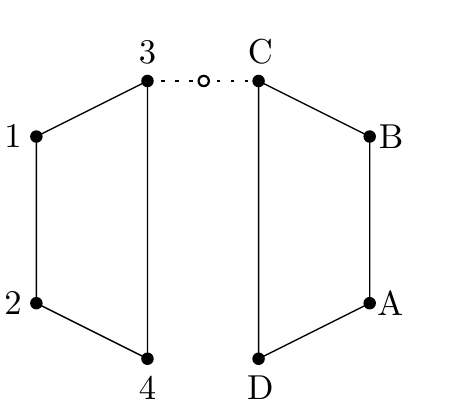} & 
      \includegraphics{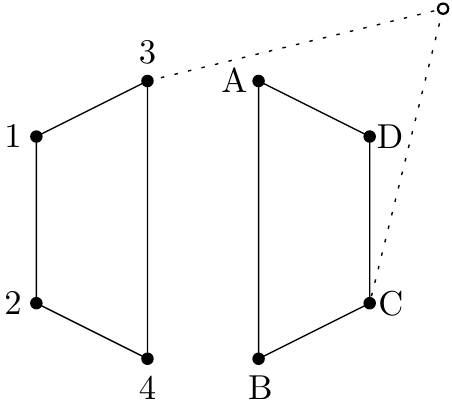} \\
      $\mathcal{G}$ & $G_1$ & $G_2$ 
    \end{tabular}
  }
  \caption{Two drawings, $G_1$ and $G_2$, of the same graph, $\mathcal{G}$, where making $\mathcal{G}$ connected requires the addition both of edges and vertices. In this case, $\mathcal G$ is made connected by adding the hollow vertex and two dashed edges.}
  \label{fig:bad-example}
\end{figure}

The motivation for this work comes from the
problem of morphing planar graphs, which has many applications
\cite{erten.kobourov.ea:intersection-free,friedrich.eades:graph,gotsman.surazhsky:guaranteed,surazhsky.gotsman:controllable,surazhsky.gotsman:intrinsic}
including computer animation.  Imagine an animator who wishes to animate a
scene in which a character's expression goes from neutral, to surprised,
to happy (see Figure~\ref{fig:faces}). The animator can draw these
three faces, but does not want to hand-draw the 30--60 frames required
to animate the change of expression.  The strokes used to draw the
character's features can be converted into paths and these can be merged
into components corresponding to the character's eyes, nose, mouth and
so on.  A correspondence between the same elements in different pictures
is also given.\footnote{In many cases, the correspondence is a byproduct
of the creation process. For example, in Figure~\ref{fig:faces}, the
second two faces were obtained by copying and then editing the first one.}
Thus, the input is three isomorphic drawings of the same planar graph.

\begin{figure}
  \centering{ 
  \begin{tabular}{c@{\hspace{.3cm}}c@{\hspace{.3cm}}c@{\hspace{.3cm}}c}
    \includegraphics{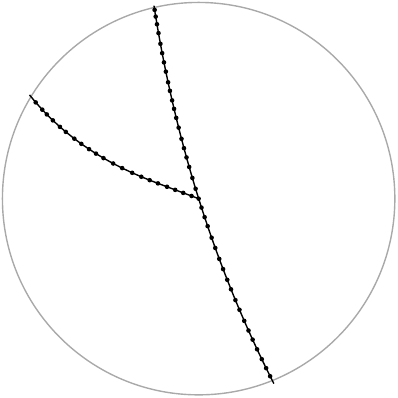} &
    \includegraphics{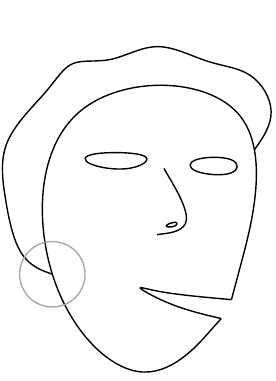} &
    \includegraphics{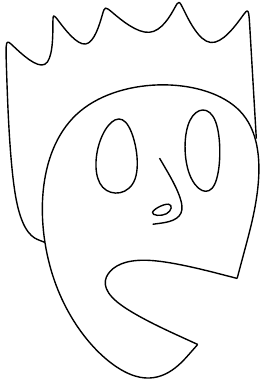} &
    \includegraphics{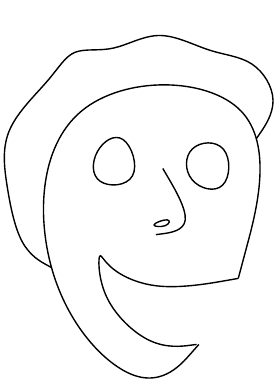} 
  \end{tabular}}
  \caption{Computer-assisted animation frequently involves morphing between
   a sequence of drawings of the same planar graph.  Zooming in on a section
   of the image reveals that the artist's strokes are approximated by polygonal
   paths}
  \label{fig:faces}
\end{figure}

In this setting, animating the face becomes a problem of
\emph{morphing} (i.e., continuously deforming) one drawing of
a planar graph into another drawing of the same planar graph
while maintaining planarity of the drawing throughout the
deformation. This morphing problem has been studied since 1944,
when Cairns \cite{cairns:deformations} showed such a transformation
always exists.  Since then, a sequence of results has shown that morphs
can be done efficiently, so that the motion can be described concisely
\cite{alamdari.angelini.ea:morphing,angelini.dalozzo.ea:morphing,grunbaum.shephard:geometry,thomassen:deformations}.
The most recent such result \cite{angelini.dalozzo.ea:morphing} shows
that any planar drawing of an $n$-vertex \emph{connected} planar graph
can be morphed into any isomorphic drawing using a sequence of $O(n)$
\emph{linear morphs}, in which vertices move along linear trajectories
at constant speed.

The morphing algorithms discussed above require that the input graph,
$\mathcal{G}$, be connected. In many applications of morphing (for example
in Figure~\ref{fig:faces}) the input graph is not connected. Before these
morphing algorithms can be used, $\mathcal{G}$ must be augmented into a
connected graph, $\mathcal H$, but this augmentation must be compatible
with the drawings of $\mathcal{G}$.  At the same time, the complexity
of the morph produced by a morphing algorithm depends on the number of
vertices of $\mathcal H$.  Therefore, we want to find an augmentation
with the fewest number of vertices.  This motivates the theoretical
question studied in the current paper.


\subsection{Formal Problem Statement and Main Result}

A \emph{drawing} of a graph $\mathcal{G}=(V,E)$ is a one-to-one
function $\psi\colon V\to\R^2$.  A drawing is \emph{planar} if (a)~for
every pair of edges $uw$ and $xy$ in $E$, the open line segment with
endpoints $\psi(u)$ and $\psi(w)$ is disjoint from the open line
segment with endpoints $\psi(x)$ and $\psi(y)$ and (b)~for every edge
$uw$ and every vertex $y$, $\psi(y)$ is not contained in the open line
segment with endpoints $\psi(u)$ and $\psi(w)$.  Two planar drawings,
$\psi_1$ and $\psi_2$, of $\mathcal{G}$ are \emph{isomorphic} if
there exists a continuous family of planar drawings $\{\psi^{(t)}
\colon 0\le t\le 1\}$ of $\mathcal{G}$ such that $\psi^{(0)}=\psi_1$
and $\psi^{(1)}=\psi_2$.\footnote{By Cairn's result, this is equivalent
to saying that the two drawings of $G$ have the same rotation schemes,
the same cycle-vertex containment relationship, and the same outer face.}

We call a graph, $G=(V,E)$ a \emph{geometric planar graph} if it is
the image of some planar drawing of a graph $\mathcal{G}=(\mathcal
V, \mathcal E)$.  That is, $V(G)=\{\psi(v): v\in V(\mathcal{G})\}$,
$E(G)=\{(\psi(u),\psi(w)) : (u,w)\in E(\mathcal G)\}$, and $\psi$ is a planar
drawing of $\mathcal G$.  When clear from context, we will sometimes
treat a geometric planar graph interchangeably with the set of points
and line segments defined by its vertices and edges, respectively.

We will avoid repeatedly referencing drawing
functions like $\psi$.  Instead, we will talk about a graph $\mathcal{G}$
and $k$ isomorphic drawings $G_1,\ldots,G_k$ of $\mathcal G$.  This
means that each $G_i$ is the geometric graph given by the drawing of
$\mathcal{G}$ with some function $\psi_i$ and that $\psi_1,\ldots,\psi_k$
are pairwise isomorphic.  When necessary, we may talk about the vertex $v$
in $G_i$ where $v$ is actually a vertex of $\mathcal{G}$; this should
be taken to mean the vertex $\psi_i(v)$ in $G_i$.  

We are now ready to state the main problem studied in this paper.
Given $k>1$ planar isomorphic drawings $G_1, \ldots, G_k$
of $\mathcal G$, a \emph{compatible augmentation}, $\mathcal H$, of
$\mathcal G$ is a supergraph of $\mathcal G$ such that (1) $\mathcal H$
is connected, and (2) there exist planar isomorphic drawings,
$H_1, \ldots, H_k$, of $\mathcal H$ such that $H_i\supset G_i$ for every
$i\in\{1,\ldots,k\}$.  We prove the following result:

\noindent\textbf{Main Result:} {\itshape If $\mathcal{G}$ is a graph
with $n$ vertices and $r\in\{2,\ldots,n\}$ connected components and
$G_1,\ldots,G_k$, $k\ge 2$, are isomorphic planar drawings of $\mathcal
G$, then there always exists a compatible augmentation of $\mathcal G$
whose size is $O(nr^{1-1/k})$.

Furthermore, this bound is tight; for every $r\in\{2,\ldots,\lfloor n/4\rfloor\}$ and $k\ge 2$, there exists a graph $\mathcal G$
with $r$ components and $k$ isomorphic planar drawings for which any compatible
augmentation has size $\Omega(nr^{1-1/k})$.}

These results show that the (worst-case) cost of an augmentation is very
sensitive to the number, $k$, of drawings, but only up to a point.
For a fixed value of $r$, our bounds range from $\Theta(nr^{1/2})$ (when
$k=2$) to $\Theta(nr)$ (for $k\ge \log r$).  On the other hand, for the
common case where $k=2$, our bounds vary from $\Theta(n)$ (when $r\in
O(1)$) up to $\Theta(n^{3/2})$ (when $r\in\Theta(n)$).  Neither $r$ nor $k$
causes the complexity of the augmentation to blow up beyond $\Theta(n^2)$.

\subsection{Related Work}

To the best of our knowledge, there is little work on
compatible connectivity-augmentation of planar graphs, though
there is work on isomorphic triangulations of polygons.  Refer to
Figure~\ref{fig:compatible-triangs}.  In this setting, the graph
$\mathcal{G}$ is a cycle and one has two non-crossing drawings, $P$
and $Q$, of $\mathcal G$. The goal is to augment $\mathcal G$ (and the
two drawings $P$ and $Q$) so that $\mathcal G$ becomes a near-triangulation,
and $P$ and $Q$ become (geometric) triangulations of the interiors
of the polygons whose boundaries are $P$ and $Q$.  Aronov \etal\
\cite{aronov.seidel.ea:on} showed that this can always
be accomplished with the addition of $O(n^2)$ vertices and that
$\Omega(n^2)$ vertices are sometimes necessary.  Kranakis and Urrutia
\cite{kranakis.urrutia:isomorphic} showed that this result can be made
sensitive to the number of reflex vertices of $P$ and $Q$, so that the
number of triangles required is $O(n+pq)$ where $p$ and $q$ are the
number of reflex vertices of $p$ and $q$, respectively.

\begin{figure}
  \centering{
    \includegraphics{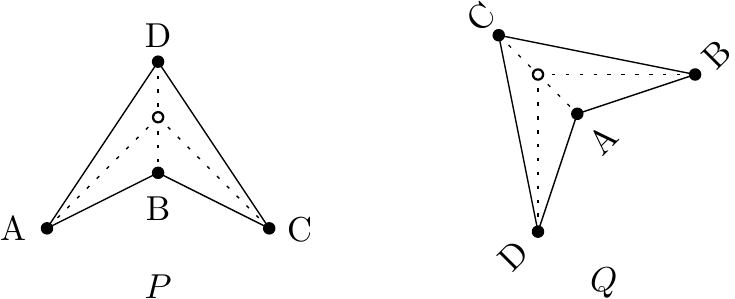}
  }
  \caption{Compatible triangulations of two 4-gons $P$ and $Q$.}
  \label{fig:compatible-triangs}
\end{figure}

Babikov \etal\ \cite{babikov.souvaine.ea:constructing} showed that the result of Aronov \etal\ can be extended to polygons with holes. This work is the most closely related to ours because it encounters (the special case $k=2$ of) our problem as a subproblem. In their setting, the graph $\mathcal G$ is a collection of $r$ cycles, the drawings $P$ and $Q$ are such that one cycle, $\mathcal C$, of $\mathcal G$ contains all the others in its interior and no other pair of cycles is nested in $P$ or $Q$. In the first stage of their algorithm, they build a connected supergraph $\mathcal{H}'$ of $\mathcal{G}$, but their supergraph has size $\Theta(n^2)$ in the worst case.  The main theorem in the current paper shows that this step of their algorithm could be done with a graph $\mathcal{H}'$ having only $O(nr^{1/2})$ edges (but completing this graph to a triangulation may still requires $\Omega(n^2)$ edges in the worst case).

Finally, several papers have dealt with the problem
of increasing the connectivity of a (single) geometric
planar graph while adding few vertices and edges.  Abellenas
\etal~\cite{abellanas.olaverri.ea:augmenting} consider the problem of
adding edges to a planar drawing in order to make it 2-edge connected
and showed that $\lfloor(2n-2)/3\rfloor$ edge are sometimes necessary
and $6n/7$ edges are always sufficient.  T\'oth \cite{toth:connectivity}
later obtained the tight upper-bound of $\lfloor(2n-2)/3\rfloor$ for
the same problem.  Rutter and Wolff \cite{rutter.wolff:augmenting} show
that finding the minimum number of edges required to achieve 2-edge
connectivity is NP-hard.

\subsection{Outline}

To guide the reader, we give a rough sketch of our upper bound proof,
which is illustrated in Figure~\ref{figure:example}.  We will assume,
for the sake of simplicity, that every component has at least one vertex
incident to the outer face.

For each component, $\mathcal C_i$, of $\mathcal G$ we select
a distinguished \emph{corner}, $a_i$, of $G_1$. (A corner is the
space between two consecutive edges incident to some vertex of the
outer face). The corner $a_i$ is called the \emph{attachment corner}
for component $\mathcal C_i$.  Notice that, since $G_1,\ldots,G_k$
are isomorphic, $a_i$ appears as a corner in each of $G_1,\ldots,G_k$.
The augmentation that we ultimately create will consist of a path that
connects to each component $\mathcal C_i$ at its attachment corner $a_i$.

\begin{figure}
  \begin{center}
   \begin{tabular}{c}
     \includegraphics{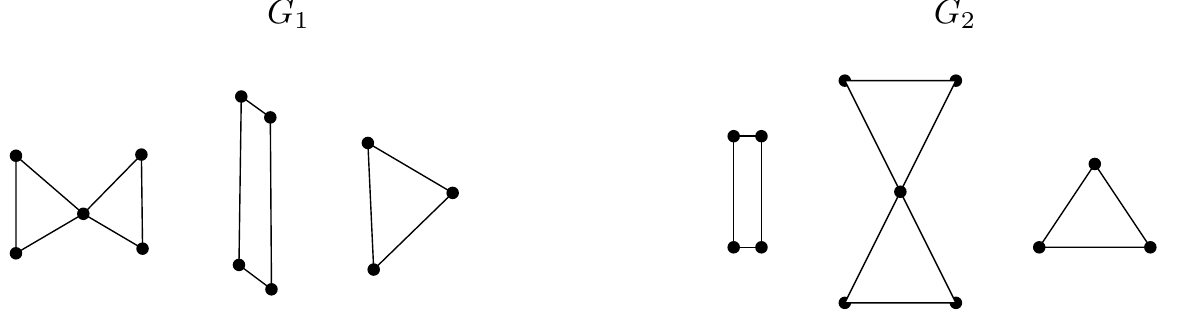} \\[2ex]
     $\Downarrow$ \\
     \includegraphics{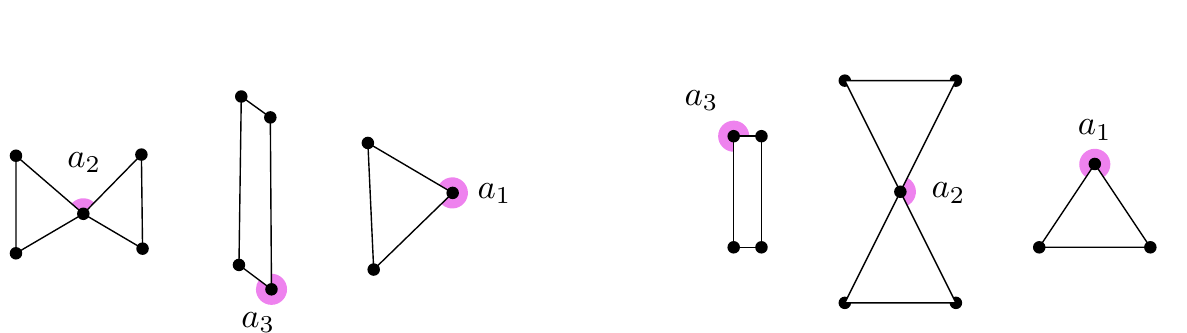} \\[2ex]
     $\Downarrow$ \\
     \includegraphics{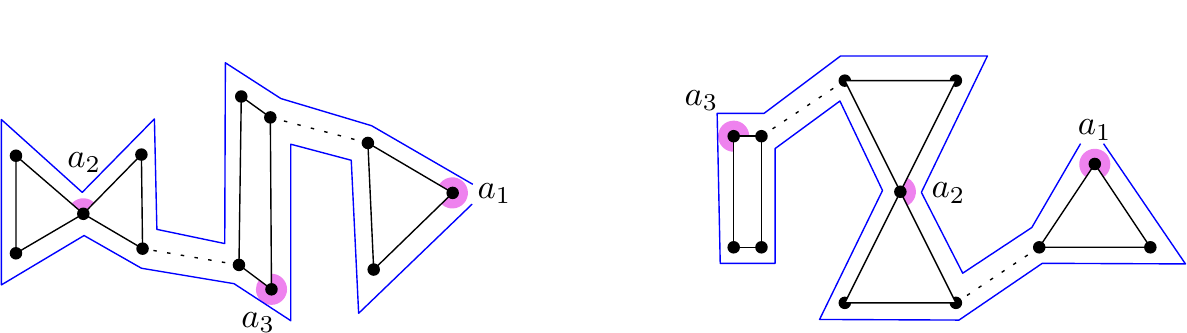} \\[2ex]
     $\Downarrow$ \\
     Hamiltonian Path Algorithm Gives Path $a_3,a_2,a_1$ \\[2ex]
     $\Downarrow$ \\
     \includegraphics{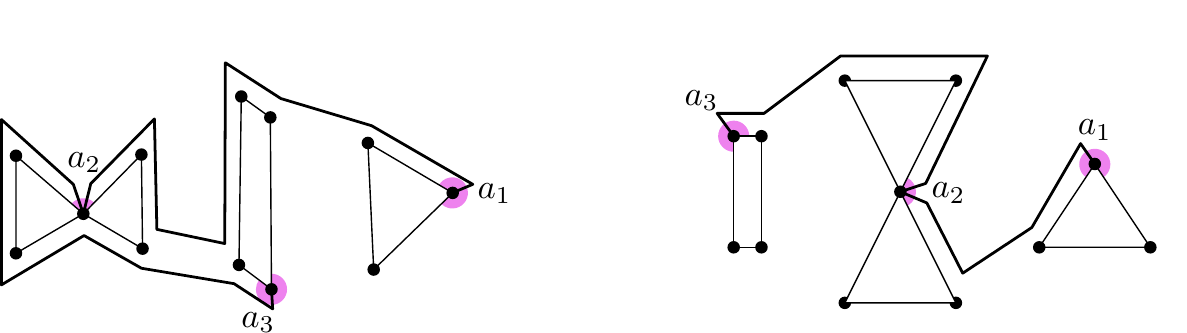} 
   \end{tabular}
  \end{center}
  \caption{The algorithm for making a compatible augmentation of $\mathcal{G}$ works by defining corners $a_1,\ldots,a_r$, taking a spanning path on each augmented drawing $G_i^*$ that visits all corners, and using this information to compute a permutation of $a_1,\ldots,a_r$ that can be drawn efficiently in each $G_i$.}
  \label{figure:example}
\end{figure}

Next, for each drawing, $G_j$, we add $r-1$ edges to make $G_j$ into a
connected graph, $G_j^*$; these edges are not, in general, edges that
take part in the final augmentation of $G_1,\ldots,G_k$ (in fact the
edges of $G_j^*$ not in $G_j$ might be different from the edges of
$G_h^*$ not in $G_h$, if $j\neq h$).  We then traverse the boundary
of the outer face of $G_j^*$ to obtain a polygonal path, $\varphi_j$,
of length $O(n)$ that comes close to every corner of $G_j$.  The path
$\varphi_j$ is then used to define an integer distance $d_j(a_\ell,a_m)$
for any two attachment corners $a_\ell$ and $a_m$. This distance includes
information about the number of edges of $\varphi_j$ between $a_\ell$
and $a_m$ as well the sizes of the some of the components visited while
walking from $a_\ell$ to $a_m$ along $\varphi_j$.

Next, we take a leap into $k$ dimensions by using the distance
functions $d_1,\ldots,d_k$ to produce a $k$-dimensional point set
$X=\{x_1,\ldots,x_r\}$ that lives in a hypercube of side-length $O(n)$.
This mapping has the property that, by adding a path of length $O(\|x_\ell
-x_m\|_\infty)$ to $\mathcal G$, the attachment corners $a_\ell$ and $a_m$
can be joined in each of $G_1,\ldots,G_k$ while preserving planarity.

Now, since the point set $X$ is in $\R^k$, has $r$ points, and lives in a
hypercube of side-length $O(n)$, a classic argument about the geometric
Travelling Salesman Problem \cite{few:shortest,moran:on} implies that it has a
spanning path whose length, measured in the $\ell_\infty$ norm, is
$O(nr^{1-1/k})$.\footnote{The original argument was for the standard
$\ell_2$ norm and has an extra $\sqrt{k}$ factor.  This factor disappears
in the $\ell_\infty$ norm.  See Appendix~\ref{app:uniform-norm} for details.}
This implies that $\mathcal G$ can be made connected with a collection
of $r-1$ paths, whose endpoints are the corners $a_1,\ldots,a_r$, having
total size $O(nr^{1-1/k})$, and that each of these paths can be
drawn in a planar fashion in each of $G_1,\ldots,G_k$.

At this point, all that remains is to show that each of these $r-1$ paths
be drawn in each of $G_1,\ldots,G_k$ without crossing
each other.  This part of the proof involves carefully winding these
paths around the components in $G_1,\ldots,G_k$ using paths close
to the paths $\varphi_1,\ldots,\varphi_k$ defined above. This part
of the proof resembles the first part of the proof of Babikov \etal\
\cite{babikov.souvaine.ea:constructing}, but is complicated by the fact
that we have to be quite careful that the number of edges in these paths
remains in $O(nr^{1-1/k})$. 


The remainder of the paper is organized as follows: In
Section~\ref{section:Trivial components} we start by solving
the special case in which the graph $G$ has no edges. This
special case is already non-trivial and introduces some of the
main ideas used in solving the full problem, which is tackled in
Section~\ref{section:General}. Section~\ref{section:Lower bound}
presents a lower bound construction that matches our upper bound.

\section{Upper bounds for trivial components}\label{section:Trivial components}
As a warmup, we consider a (trivial) graph containing $n$ vertices and no edges.
Before constructing a compatible augmentation, we provide a subroutine
that constructs a ``short'' planar spanning path of a given ordered set of points.

\subsection{Spanning paths of point sets}
Let $S$ be a set of $n$ points in the plane with distinct $x$-coordinates. 
Given a point $v\in S$, let $\rank(v)$ denote the number of points of $S$ that lie to the left of (having smaller $x$-coordinate than) $v$.

Given an arbitrary order $(v_1, v_2, \ldots, v_n)$ of the points of $S$, we want to construct a path $R$ that connects them in this order and such that:  
\[
   |R|  = O\left(\sum_{i=1}^{n-1} |\rank(v_i) - \rank(v_{i+1})| \right).
\]
This paper uses the $|\cdot|$ operator in several different ways, depending on the type of its argument.  For a real number, $x$, $|x|$ is the absolute value of $x$.  For a walk, $R=(r_0,\ldots,r_k)$, $|R|=k$ denotes the number of edges traversed by $R$. For a (weakly-)simple polygon, $P$ whose vertices---as encountered during a counterclockwise traversal---are $(a_1,\ldots,a_k)$, $|P|=k$, denotes the number of edges of $P$.

Consider a horizontal line $\ell$ below $S$ and let $\pi$ be the closed halfspace supported by $\ell$ that contains $S$.  We present an algorithm that constructs $R$ iteratively; during the $i$th iteration of the algorithm, the path is extended with $O(|\rank(v_i) - \rank(v_{i-1})|)$ vertices to include $v_i$.  For each $i\in \{1,\dots,n\}$, after the $i$th iteration of the algorithm, we maintain the invariant that $\ell$ does not intersect $R$, and we also maintain the \emph{escape invariant} which is defined as follows:
For each $j\in \{i, \dots, n\}$, there is a closed cone $\Delta_{j}$ with apex $u_{j}$ such that (1) $u_j$ lies above $v_j$ and has the same $x$-coordinate as $v_j$ ($u_j =v_j$ if $j = i$), (2) $\Delta_j$ contains $v_j$ and no other point of $S$, (3) $\Delta_j$ contains the ray originating at $v_j$ in the direction of the negative $y$-axis, (4) $\Delta_j$ does not intersect $R$, and (5) $\Delta_h$ and $\Delta_j$ are disjoint inside $\pi$, for every $h\in \{i,\dots,n\}$ with $h\neq j$.

Initialize $R$ as a path that consists of the single vertex $v_1$. In order to establish the escape invariant, we define $u_1=v_1$; also, for each $j\in \{2,\dots,n\}$, we define $u_j$ as an arbitrary translation up of $v_j$; further, for each $j\in \{1,\dots,n\}$, we let $\Delta_j$ be a cone with apex on $u_j$ sufficiently narrow so that these cones do not intersect inside $\pi$; see Figure~\ref{fig:Escape Invariant}.

\begin{figure}[tb]
\centering 
\includegraphics{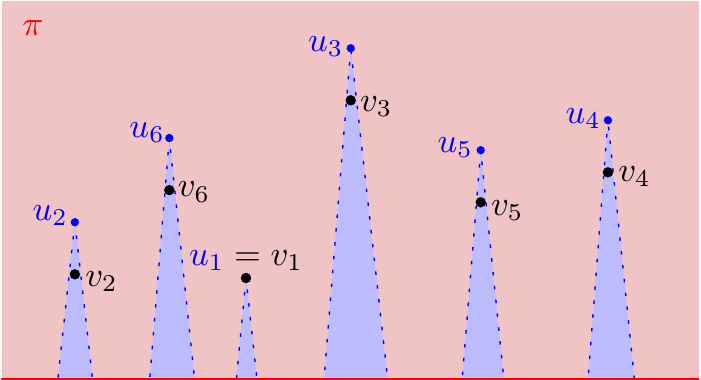}
\caption{The halfplane $\pi$ and the cones $\Delta_1,\ldots,\Delta_n$ with apexes at $u_1,\ldots,u_n$.}
\label{fig:Escape Invariant}
\end{figure}


Now assume that $R$ is a path connecting $v_1$ with $v_j$, for some $j\in \{1,\dots,n-1\}$. We extend $R$ by appending a
path that connects $v_j$ with $v_{j+1}$.

First, we translate $\Delta_{j+1}$ down until its apex $u_{j+1}$ coincides with $v_{j+1}$. Let $\pi_j$ be the closure of the set obtained from $\pi$ by removing $\Delta_h$, for every $h\in\{j,j+1,\ldots,n\}$; see Figure~\ref{fig:Dented Halfspace} (right). That is, $\pi_j$ is a halfspace with dents made by the removal of $n-j+1$ cones.
Observe that, for every pair of apexes $u_i$ and $u_h$, with $i,h\geq j$, the boundary of $\pi_j$
contains a path from $u_i$ to $u_h$ with $O(|\rank(v_i)-\rank(v_h)|)$ edges.
Because $\ell$ does not intersect $R$ and by the escape invariant, the boundary of $\pi_j$ intersects $R$ only at $v_j$. Moreover, again by the escape invariant, for each $v_i$ with $i > j$, $v_i$ lies outside of~$\pi_j$ except for $v_{j+1}$ that lies on its boundary. Because both $v_j$ and $v_{j+1}$ lie on the boundary of $\pi_j$, which does not intersect $R$ other than at $v_j$, we can connect $v_j$ with $v_{j+1}$ via a path contained in the boundary of~$\pi_j$ with length $O(|\rank(v_j) - \rank(v_{j+1})|)$. In this way, we extend $R$ to a planar path that contains $v_1,\ldots,v_{j+1}$.

\begin{figure}[tb]
\centering
\includegraphics[width=.98\textwidth]{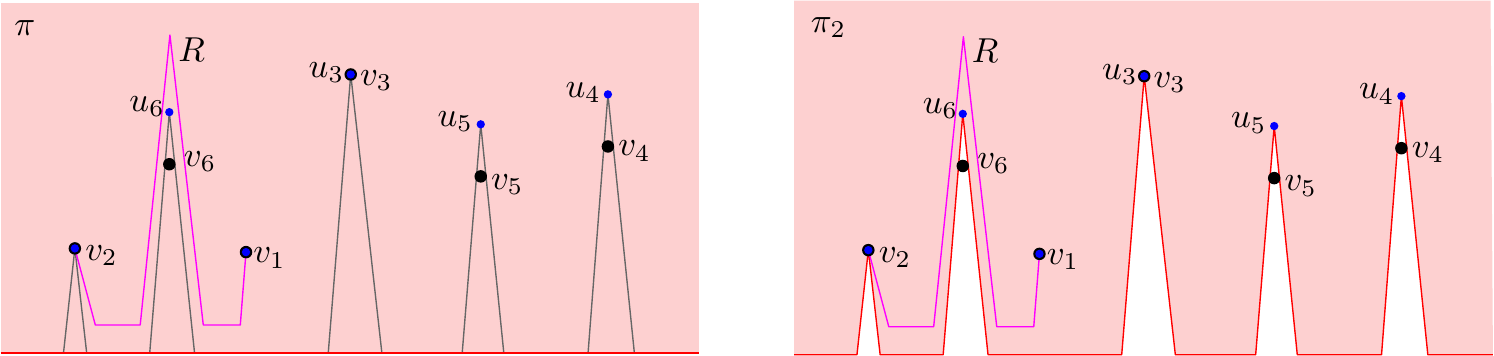}
\caption{The boundary of $\pi_2$ is in the path from $v_2$ to $v_3$.}
\label{fig:Dented Halfspace}
\end{figure}

After connecting $v_j$ with $v_{j+1}$, for each $h\in\{j{+}2,\ldots,n\}$, either $\Delta_h$ is disjoint from $R$, or it shares some portion of its boundary with $R$. However, the interior of $\Delta_h$ does not intersect $R$.
To preserve the escape invariant, for each $h\in\{j{+}2,\ldots,n\}$, we translate $\pi$ and $\Delta_h$ downwards by a sufficiently small amount, $\varepsilon$, and we scale $\Delta_{j+1}$ horizontally down, while keeping its apex at $v_{j+1}$. To conclude, each translated or scaled cone is contained in the previous one, $u_h$ lies above $v_h$, for each $h\in\{j+2,\ldots,n\}$, and $u_{j+1}$ coincides with $v_{j+1}$. Therefore, by choosing $\varepsilon$ sufficiently small, we maintain the escape invariant and obtain the following result.

\begin{lemma}\label{lemma:Compatible augmentation for trivial components}
Given an order $(v_1, \ldots, v_n)$ of the vertices of $S$, there exists a planar path $R$ that connects every point of $S$ in the given order such that the number of vertices of $R$ between $v_{i}$ and $v_{i+1}$ is $O(|\rank(v_i) - \rank(v_{i+1})|)$, for each $i\in \{1,\dots,n-1\}$.
\end{lemma}
\begin{proof}
Recall that in each iteration, the algorithm computes a path connecting $v_j$ with $v_{j+1}$ that does not cross the portion of the path already constructed. Because this invariant is maintained throughout, the resulting path is planar.

Since the path that connects $v_j$ with $v_{j+1}$ follows the boundary of $\pi_j$ and since this boundary has length $O(|\rank(v_j) - \rank(v_{j+1})|)$ between $v_j$ and $v_{j+1}$, the path that connects $v_j$ with $v_{j+1}$ has length $O(|\rank(v_j) - \rank(v_{j+1})|)$. Consequently,  the total length of $R$ is given by $O\left(\sum_{i=1}^{n-1} |\rank(v_i) - \rank(v_{i+1}) |\right)$.
\end{proof}

\subsection{Compatible drawings of point sets}

Recall that in this section $\mathcal G$ is a graph with $n$ trivial components.
Let $G_1, \ldots, G_k$ be $k>1$ isomorphic drawings of $\mathcal G$, i.e., $G_i$ is a sequence of $n$ points in the plane.
Assume without loss of generality that no two points of $G_i$ share the same $x$-coordinate.
Given a vertex $v$ of $\mathcal G$, let $\rank_{G_i}(v)$ denote the number of points of $G_i$ having smaller $x$-coordinate than $v$, and let $x_v = (\rank_{G_1}(v), \ldots, \rank_{G_k}(v))$ be a point in the integer grid $[0;n-1]^k$ in $\mathbb{R}^k$.  Let $X = \{x_v : v\in V(\mathcal G)\}$ and let $P$ be the shortest Hamiltonian path of $X$ when distance is measured using the $\ell_\infty$ norm, so that
\[
   \|x_v-x_u\|_\infty = \max\{|\rank_{G_i}(v)-\rank_{G_i}(u)|:i\in\{1,\ldots,k\}\} \enspace .
\]
It is known that the length of $P$ is $O(n^{2-1/k})$ (see
Corollary~\ref{cor:tsp} in Appendix~\ref{app:uniform-norm}).  Note that
the order of the points of $P$ induces an order on the vertices of
$\mathcal G$ and hence, an order on the vertices of each $G_i$.

\begin{theorem}\label{theorem:points}
For each $i\in \{1,\dots,n\}$, we can construct a path $R_i$ of length $O(n^{2-1/k})$ that connects every point of $G_i$ so that $G_i\cup R_i$ is planar. Moreover, for any distinct $i,j\in \{1,\dots,n\}$, $G_i\cup R_i$ and $G_j\cup R_j$ are~isomorphic.
\end{theorem}
\begin{proof}
By relabelling, let $(v_1, \ldots, v_n)$ denote the order of the vertices of $\mathcal{G}$ induced by the Hamiltonian path, $P$, described above.  The graph $\mathcal{H}$, which is an augmentation of $\mathcal{G}$, is a path that visits the vertices $v_1,\ldots,v_n$ in this order. Letting $d_j$ denote the $\ell_\infty$ distance between $x_{v_j}$ and $x_{v_{j+1}}$, the path $\mathcal{H}$ includes an additional $O(d_j)$ vertices between $v_{j}$ and $v_{j+1}$.  It follows that the number of vertices
in $\mathcal{H}$ is proportional to the length of $P$, which is $O(n^{2-1/k})$.

For each $G_i$, we use Lemma~\ref{lemma:Compatible augmentation for trivial components} to draw $\mathcal{H}$ as a planar path, $R_i$,
that connects the vertices $v_1,\ldots,v_n$ in this order in the drawing $G_i$.
Since $d_j\ge |\rank_{G_i}(v_j) - \rank_{G_i}(v_{j+1})|$, the $O(d_j)$
vertices in $\mathcal{H}$ between $v_j$ and $v_{j+1}$ are enough to
draw the $O(|\rank_{G_i}(v_j) - \rank_{G_i}(v_{j+1})|)$ vertices in $R_i$
between $v_j$ and $v_{j+1}$.
Since the vertices of each $G_i$ are connected in the same order,
$G_i\cup R_i$ is isomorphic to $G_j\cup R_j$ for each $i,j\in\{1,\ldots,k\}$.
\end{proof}

\section{The general problem}\label{section:General}
In this section, we extend the result presented in
Section~\ref{section:Trivial components} to graphs with
non-trivial components.  We follow the same general scheme used in
Section~\ref{section:Trivial components} for the case of trivial
(isolated vertex) components:  We define $k$ different
orderings of the components of $\mathcal G$ and use these orderings (and
the sizes of these components) to define an $r$-point set, $X$, in $\R^k$. A
short path that visits all points in $X$ is then translated back into
a short path, $R$, that visits all components of $\mathcal G$. The path
$R$ is then
added, as a polygonal path, $R_i$, to each drawing, $G_i$, of $\mathcal G$.

Unlike the case in which
components are isolated vertices, there is no natural ordering of the
components of $G_i$, so we must define one. Also, the drawing
of path $R_i$ is considerably more complicated.  In Section~\ref{section:Trivial components}, $R_i$ is drawn incrementally, and always passes above components that are not yet included in $R_i$ and below components that are already included in $R_i$.  In this section, we redefine ``above'' and ``below''. The number of edges required to go above or below a component depends on its size and structure.

\subsection{Preliminaries}\label{section:Preliminaries} 
Let $C$ be a connected geometric planar graph. Let $v_0, v_1, \ldots, v_k, v_0$ be the sequence of vertices of $C$ visited by a counterclockwise
Eulerian tour along the boundary of the outer face of $C$. Note that
$v_i$ may be equal to $v_j$ for some $i\neq j$.  A vertex $v_i$
in this sequence is called a \emph{corner} of $C$.  We consider the boundary of $C$, denoted by $\partial C$, to be the
boundary of the weakly-simple polygon $(v_0, \ldots, v_k, v_0)$ whose
vertex set is the set of corners of $C$.\footnote{More formally, $\partial C$ is the boundary of the unbounded component of $\mathbb{R}^2\setminus C$, when we treat $C$ as the union of all its edges (line segments) and vertices (points).}

Let $\varepsilon >0$. For each corner $v_i$ of $\partial C$, let $\ell_i$ be the half-line starting at $v_i$ that bisects the angle between the edges $v_{i-1}v_i$ and $v_i v_{i+1}$ in the outer face of $C$. Let $z_i$ be the point at distance $\varepsilon$ from $v_i$ along $\ell_i$. We call $z_i$ the \emph{$\varepsilon$-copy} of $v_i$. Let $\partial_\varepsilon C$ be the piecewise-linear cycle defined by the sequence $(z_0, z_1, \ldots, z_k, z_0)$. We call $\partial_\varepsilon C$ the \emph{$\varepsilon$-fattening} of $C$.
An $\varepsilon$-fattening $\partial_\varepsilon C$ is \emph{simple} if $\partial_\varepsilon C$  is a simple polygon that contains~$C$.
Note that $\partial_\varepsilon C$ is simple, provided that $\varepsilon$ is sufficiently small. In this paper, we consider only simple $\varepsilon$-fattenings; see Figure~\ref{fig:Blowing}. Note that the number of edges between two corners of $\partial C$ along the boundary of $C$ is the same as the number of edges between their $\varepsilon$-copies along~$\partial_\varepsilon C$.

\subsection{Connected augmentations}\label{section: connected augmentations}
Let $G$ be a geometric planar graph with $r$ connected components such that each component is adjacent to the outer face.
Two vertices are {\em visible} if the open segment joining them does not intersect $G$.
Let $T_G$ be a smallest set of edges of the visibility graph of $G$ that need to be added to $G$ to make it connected.
As there are always two components containing mutually visible vertices, we can connect them and repeat recursively.  Thus, $T_G$ has $r-1$ edges. (Loosely, we can think of $T_G$ as a spanning tree of $G$'s components.) Let $G^* = G\cup T_G$.  We say that $G^*$ is a \emph{connected augmentation} of $G$; see Figure~\ref{fig:Blowing}.

\begin{figure}[tb]
\centering
\includegraphics{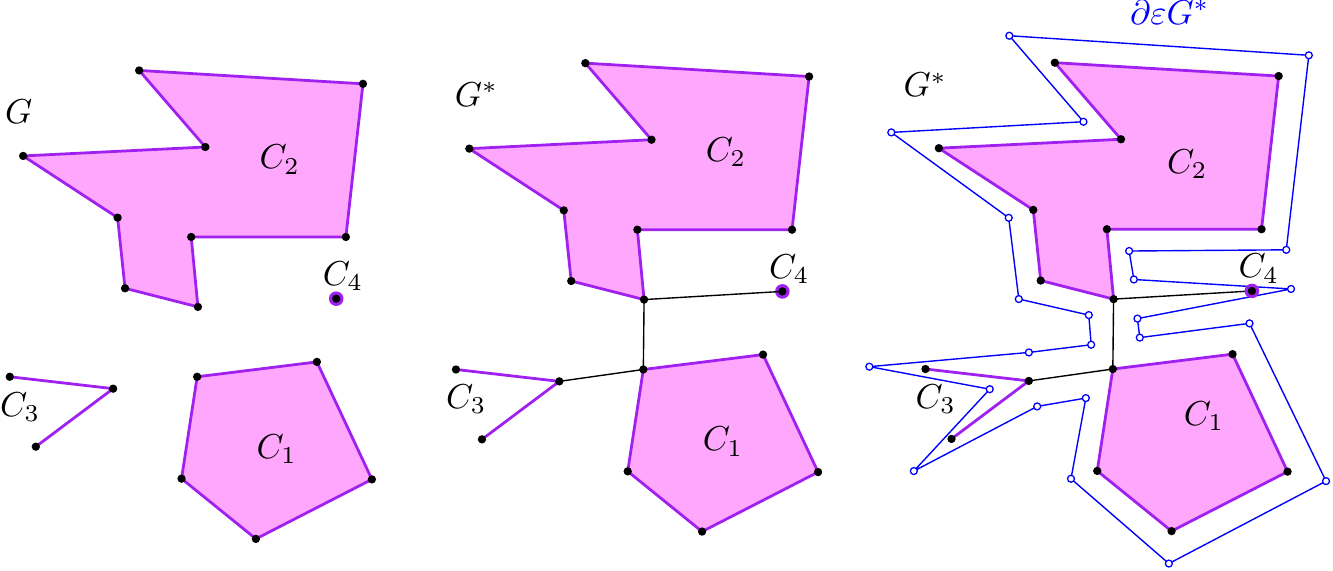}
\caption{The graph $G$ (left); a connected augmentation, $G^*$, of $G$ (middle); and the $\varepsilon$-fattening, $\partial_\varepsilon G^*$ (right).}
\label{fig:Blowing}
\end{figure}

Let $C_1, \ldots, C_r$ be the components of $G$.
For each $i\in \{1,\dots,r\}$, let $a_i\in C_i$ be an arbitrary corner of $\partial C_i$ (note that $a_i$ is adjacent to the outer face).
We call $a_i$ the \emph{attachment corner} of $C_i$.

Let $\varphi$ be the path on the corners of $\partial G^*$ (hence $\varphi$ is also a walk on the vertices of $\partial G^*$) obtained by splitting $\partial G^*$ at the corner $a_1$. That is, $\varphi$ is a path that visits every corner of $\partial G^*$ exactly once except for $a_1$, which is visited twice.
Given two corners $u$ and $v$ in $\partial G^*$, let $\varphi(u,v)$ denote
the unique path in $\varphi$ that connects $u$ with $v$. Let $A(u,v)$ be
the set of attachment corners of $G$ visited by $\varphi(u,v)$. For two attachment corners $u=a_i$ and $v=a_j$, define
\[
       \sigma_G(u,v) = |\varphi(u,v)| + 2\sum_{a_t\in A(u,v)}|\partial C_t|
           - |\partial C_i| - |\partial C_j| ,
\]
which we call the \emph{cost} of going from $u$ to $v$.  The definition of $\sigma_G(u,v)$ is designed to capture the fact that, if $u$ and $v$ occur consecutively on the path $R$ we construct, then the portion of $R$ between $u$ and $v$ will have length at least $|\varphi(u,v)|$ since it follows $\varphi$, and it may also take a detour around every attachment corner $a_i\in A(u,v)$.  If it takes this detour at some $a_i\in A(u,v)$, then it does so by walking around $\partial C_i$ which requires an additional $|\partial C_i|$ edges.

\begin{lemma}\label{lemma:Contained in integer grid}
    If $u=a_i$ is an attachment corner of $G$, then $\sigma_G(a_1, u) < 6n$. Moreover, if $v=a_j$, with $j> i$, is another attachment corner of $G$, then
  $\sigma_G(u, v) = \sigma_G(v, u) = \sigma_G(a_1, v)- \sigma_G(a_1, u)$.
\end{lemma}
\begin{proof}
Recall that $G^*$ is a graph with $n$ vertices, so $\partial G^*$
is a weakly-simple polygon with at most $n$ distinct vertices, so
$|\partial G^*|\le 2n-2<2n$.  Furthermore, $\varphi(a_1,u)\subset
\partial G^*$, so $|\varphi(a_1,u)|\le |\partial G^*|< 2n$.  Similarly,
$2\sum_{a_i\in A(a_1,u)}|\partial C_i| \le 2\sum_{i=1}^r |\partial C_i| < 4n$.
Therefore, $\sigma_G(a_1,u)< 6n$, which proves the first part of the lemma.

To prove the second part of the lemma, first observe that $\sigma_G(u,v)=\sigma_G(v,u)$ by definition. To prove the second equality, denote
the relevant attachment corners of $G^*$ by
$a_1,a_2,\ldots,a_i=u,a_{i+1},\ldots,a_j=v$. Then $A(a_1,u)
= \{a_1,\ldots,a_i\}$, $A(a_1,v)=\{a_1,\ldots,a_j\}$, and
$A(u,v)=\{a_i,\ldots,a_j\}$, so
\begin{align*}
    \sigma_{G}(a_1,v)-\sigma_G(a_1,u) 
        & = |\varphi(a_1,v)|-|\varphi(a_1,u)| + 2\sum_{t=1}^j |\partial C_t|
         - 2\sum_{t=1}^i |\partial C_t|  \\
               & \quad {} - |\partial C_1|+|\partial C_1| 
                - |\partial C_j|+|\partial C_i| \\
        & = |\varphi(u,v)| + 2\sum_{t=i+1}^j |\partial C_t| - |\partial C_j| + |\partial C_i| \\
        & = |\varphi(u,v)| + 2\sum_{t=i}^j |\partial C_t| - |\partial C_j| - |\partial C_i| \\
        & = |\varphi(u,v)| + 2\sum_{a_t\in A(u,v)} |\partial C_t| - |\partial C_j| - |\partial C_i| \\
        & = \sigma_G(u,v)  \enspace . \qedhere
\end{align*}
\end{proof}

\subsection{Spanning paths for connected augmentations}\label{section:Spanning paths for connected augmentations}
Let $a_1, \ldots, a_r$ be an arbitrary order of the attachment corners of $G$ (we can get the incremental indexing by relabeling the components). 
Given a path $R = (\rho_1, \rho_2, \ldots, \rho_t)$ that passes through all attachment corners of $G$, we say that $C_i$ \emph{lies to the right of} $R$ if (1) $a_i$ is the only vertex of $C_i$ that belongs to $V(R)$, and (2) if $a_i = \rho_j$ for some $j\in\{1,\ldots,t\}$, then $\rho_{j-1}$ and $\rho_{j+1}$ appear as consecutive vertices when sorting---in the graph $G\cup R$---the neighbors of $a_i$ in clockwise order around $a_i$ (see Figure~\ref{figure:right-of}).

\begin{figure}
\centering
  \includegraphics{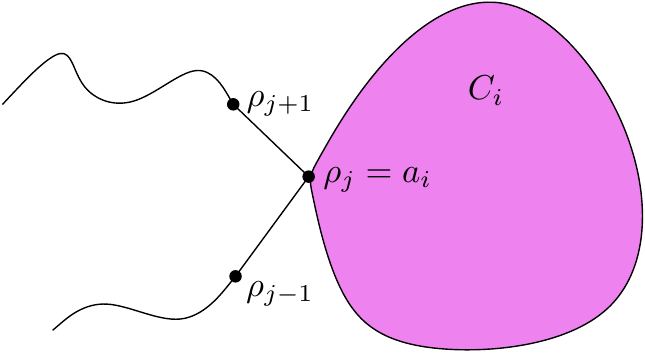}
  \caption{The component $C_i$ is to the right of the path $R=(\rho_1,\ldots,\rho_t)$.}
   \label{figure:right-of}
\end{figure}

We now show how to construct a path $R$ that connects the attachment corners of $G$ in the given order, i.e., if $i < j$, then $a_i$ is visited before $a_j$ by $R$. We want to construct $R$ so that each component $C_i$ of $G$ lies to the right of~$R$ and so that $G\cup R$ is a planar geometric graph.
Moreover, we want the subpath of $R$ between $a_j$ and $a_{j+1}$ to have $O(\sigma_G(a_j, a_{j+1}))$ vertices. 
We initialize $R$ with the trivial path that contains only $a_1$, and then extend $R$ iteratively, so that each new corner $a_i$ is included in $R$.
Recall that  for any given $\varepsilon >0$, $\partial_\varepsilon G^*$ denotes the $\varepsilon$-fattening of $G^*$ (see Section~\ref{section:Preliminaries}).
Let $\mu>0$ be a small constant to be specified later.
Initially, let $\varepsilon = 2\mu$ and let $\delta = \mu/2$. Let $\lambda < \mu/2^{r+1}$ be a constant sufficiently small so that $\partial_\lambda C_i \cap \partial_\lambda C_j = \emptyset$ for any distinct $i,j\in\{1,\ldots,r\}$.
Throughout, $\lambda$ remains constant while $\varepsilon$ and $\delta$ are redefined at each iteration. However, as an invariant we maintain $\lambda < \delta < \varepsilon$.

For each $i\in \{1,\dots,r\}$, let $w_i$ be the $\varepsilon$-copy of  $a_i$.
Split $\partial_\varepsilon G^*$ at $w_1$, i.e., $\partial_\varepsilon G^*$ is a path with both endpoints equal to $w_1$.
By choosing $\varepsilon$ sufficiently small, we guarantee that $\partial_\varepsilon G^*$ is simple, i.e., $\partial_\varepsilon G^*$ is isomorphic to $\varphi$.
We say that two points in the plane are \emph{$R$-visible} if the open segment joining them does not intersect $R$.
Let $\tau >0$. For each $i\in \{1,\dots,r\}$ such that $a_i$ is not an interior point of $R$, consider the set of points $N_i\subset \partial_\varepsilon G^*$ that are at distance at most $\tau$ from $w_i$.
Let $\Delta_i$ be the convex hull of $N_i\cup \{a_i\}$, i.e., $\Delta_i$ is a ``cone'' with apex at $a_i$; see Figure~\ref{fig:Neighborhood}. (We deliberately misuse the word ``cone'' here because the ``cones'' $\Delta_1,\ldots,\Delta_r$ in this section play the same roles as the cones $\Delta_1,\ldots,\Delta_n$ in Section~\ref{section:Trivial components}.)

While constructing $R$, we also maintain the \emph{escape invariant} which is defined as follows. Assume that $R$ so far connects $a_1,\ldots,a_i$, for some $i\in \{1,\dots,r-1\}$. Then: (1)~$R$ intersects neither $\partial_\varepsilon G^*$ nor its unbounded face; (2)~for each $j\in\{i+1,\ldots,r\}$, $R$ intersects neither the simple polygon bounded by $\partial_\delta C_j$ nor the cone $\Delta_j$; (3)~$\Delta_h\cap \Delta_j = \emptyset$, for any distinct $h,j\in\{1,\ldots,r\}$; and (4)~$w_i$ is $R$-visible from $a_i$.

In particular, Conditions~(2) and (4) of the escape invariant imply that, for
each $j\in\{i,\ldots,r\}$, every point in $N_j$ (including $w_j$) is
$R$-visible from $a_j$.  The escape invariant holds when $R=\{a_1\}$,
provided that $\tau$ is sufficiently small.

\begin{figure}[tb]
\centering
\includegraphics{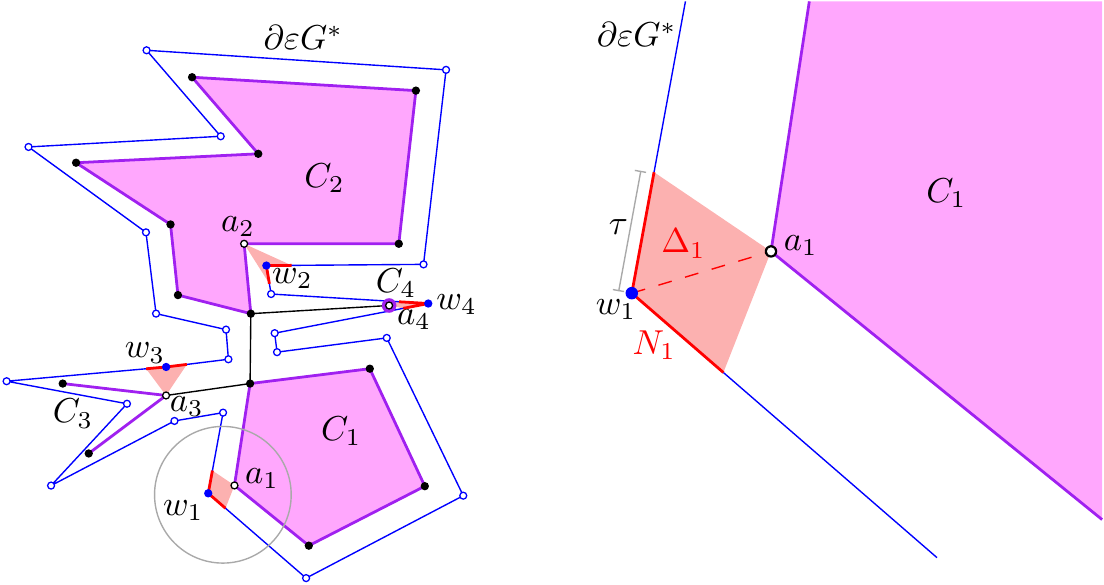}
\caption{The $\epsilon$-fattening of $G^*$ and the ``cones'' $\Delta_1,\ldots,\Delta_r$.}
\label{fig:Neighborhood}
\end{figure}

Assume that we have constructed a path $R$ that connects $a_1$ with $a_j$, for some $j\in\{1,\ldots,r-1\}$, and that the escape invariant holds.  To extend $R$, we create a new path that connects $a_j$ with $a_{j+1}$ without crossing $R$ while maintaining the escape invariant.  Recall that we consider $\partial_\varepsilon G^*$ to be a path with both endpoints on~$w_1$.

The first part of the path connecting $a_j$ with $a_{j+1}$ consists of a path connecting $a_j$ with $w_j$. If $j=1$, or if $j>1$ and $R$ together with the edge $a_j w_j$ leaves $C_j$ to its right, then connect $a_j$ with $w_j$ via a straight-line segment; since $w_j$ is $R$-visible from $a_j$, this segment does not cross $R$. Otherwise, connect $a_j$ with $\partial_\lambda C_j$ via a straight-line segment and traverse $\partial_\lambda C_j$ clockwise before moving to $w_j$ on $\partial_\varepsilon G^*$.  By the escape invariant, no crossing occur in this drawing.  In this way, we guarantee that $C_j$ lies to the right of the constructed path; see Figure~\ref{fig: Component to the right} for an illustration. Because $\lambda < \delta < \varepsilon$ and since $a_i\in V(R)$, the escape invariant is preserved.

\begin{figure}[tb]
\centering
\includegraphics[width=1\textwidth]{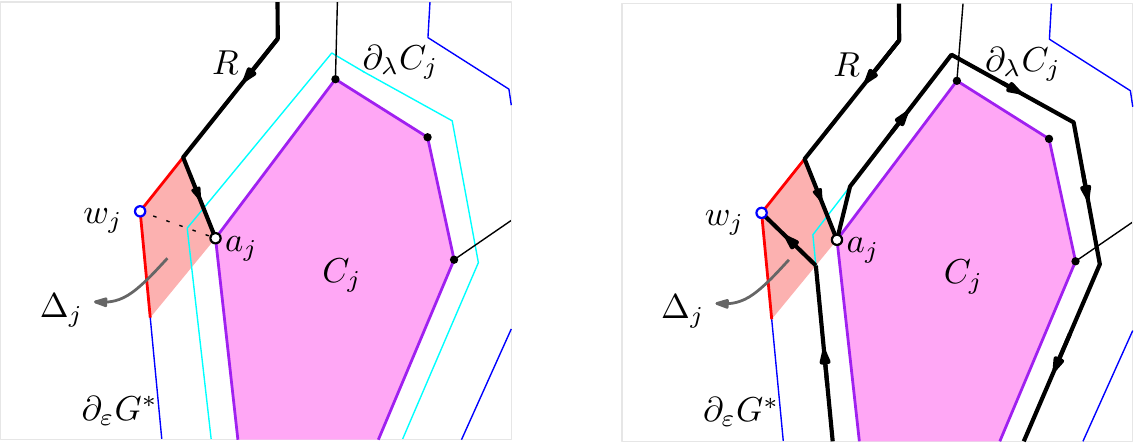}
\caption{When extending $R$ from $a_j$ to $a_{j+1}$ we have to take care to
    keep $C_j$ to the right of $R$.}
\label{fig: Component to the right}
\end{figure}

The path from $a_j$ to $a_{j+1}$ continues with a path from $w_j$ to $w_{j+1}$, which follows the unique path in $\partial_\varepsilon G^*$ from $w_j$ to $w_{j+1}$. However, whenever we reach an endpoint of $N_i$ for some $i\in \{j+2,\dots,r\}$, we take a \emph{detour} to the other endpoint of $N_i$ while avoiding its interior so that the points in the interior of $N_i$ remain $R$-visible from $a_i$; see Figure~\ref{fig:Skip Component}. Formally, we walk from the reached endpoint of $N_i$ to $\partial_\delta C_i\setminus \Delta_i$ along the boundary of $\Delta_i$. Then, we traverse the path $\partial_\delta C_i\setminus \Delta_i$ before moving to the other endpoint of $N_i$ from the endpoint of $\partial_\delta C_i \setminus \Delta_i$. Note that $R$ does not intersect the interior of the simple polygon bounded by $\partial_\delta C_i$ nor the interior of $\Delta_i$. Moreover, $R$ remains inside the simple polygon bounded by $\partial_\varepsilon G^*$.

\begin{figure}[tb]
\centering
\includegraphics[width=.98\textwidth]{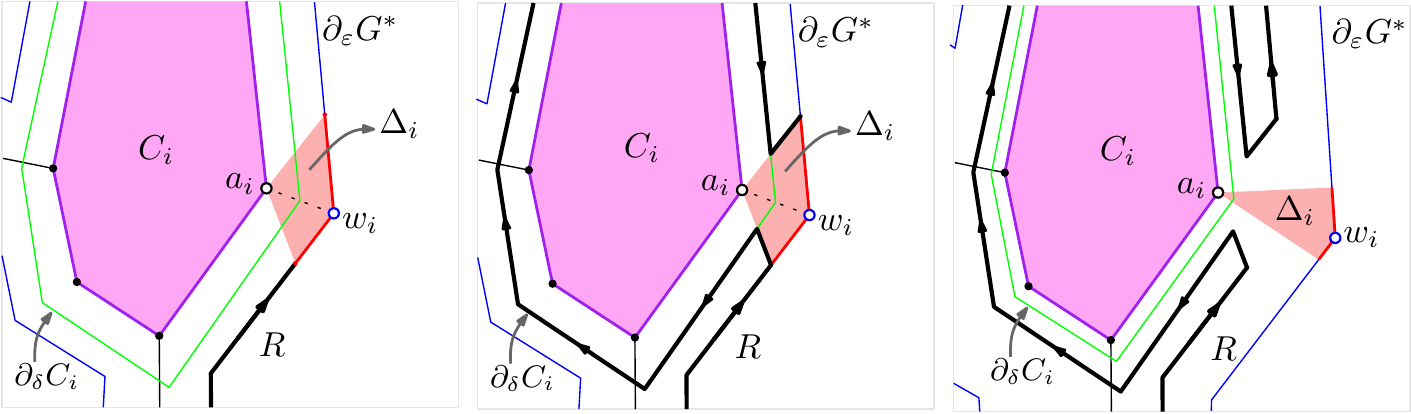}
\caption{The ``detour'' taken to avoid crossing the cone $\Delta_i$ (left, middle); and the narrowing of the cone $\Delta_i$ as well as the redefinition of the $\varepsilon$- and $\delta$-fattenings of $G^*$ and $C_i$, respectively.}
\label{fig:Skip Component}
\end{figure}

Once we go around $C_i$, we are back on $\partial_\varepsilon G^*$ on the other endpoint of $N_i$. In this way, we continue going towards~$w_{j+1}$ along $\partial_\varepsilon G^*$ until reaching an endpoint of $N_{j+1}$.
Once we reach an endpoint of $N_{j+1}$, we move directly from this endpoint to $a_{j+1}$.

Because $\partial_\varepsilon G^*$ is isomorphic to $\varphi$, the constructed path between $a_j$ and $a_{j+1}$ has length at most $|\varphi(a_j, a_{j+1})|$ plus the length of the boundaries of the components for which the path detoured. Because each component we walked around has its attachment corner on the path $\varphi(a_j, a_{j+1})$, and thus in $A(a_j, a_{j+1})$, the length of the constructed path between $a_j$ and $a_{j+1}$ is
\[ O\left(|\varphi(a_j, a_{j+1})| + \sum_{a_i\in A(a_j, a_{j+1})} |\partial C_i|\right) = O(\sigma_G(a_j, a_{j+1})) \enspace .
\]

After reaching $a_{j+1}$, we increase $\varepsilon$ by a factor of two. Similarly, we decrease the value of $\delta$ by a factor of two. That is, after reaching $a_{j+1}$, $\varepsilon = \mu 2^{j+1}$ while $\delta = \mu/2^{j+1}$ and hence, we guarantee that $\lambda < \delta < \varepsilon$.
Also, $\partial_\varepsilon G^*$ is still simple, provided that $\mu$ is initially chosen to be sufficiently small.
Finally, we reduce $\tau$ by a factor of two and update $N_i$ and $\Delta_i$ accordingly, for each $i\in \{1,\dots,n\}$; see Figure~\ref{fig:Skip Component} (c).

Recall that for each $a_i\notin R$, $R$ intersected neither the interior of $\Delta_i$ nor the interior of the polygon bounded by $\partial_\delta C_i$. Moreover, $R$ remained within $\partial_\varepsilon G^*$.
Therefore, after increasing (\emph{resp.} reducing) $\varepsilon$ (\emph{resp.} $\delta$), we preserve the escape invariant for the next iteration of the algorithm.
We iterate until all attachment corners of $G$ are visited by $R$.

\begin{lemma}\label{lemma:Path for connected augmentations}
Given an arbitrary order $a_1, \ldots, a_r$ of the attachment corners of $G$, there is a path $R$ connecting all attachment corners of $G$ in the given order such that $R\cup G$ is planar, every component $C_i$ of $G$ lies to the right of $R$ when oriented from $a_1$~to~$a_r$, and the subpath of $R$ between $a_j$ and $a_{j+1}$ has $O(\sigma_G(a_j, a_{j+1}))$ vertices, for each $j\in \{1,\dots,r-1\}$.
\end{lemma}
\begin{proof}
By construction, the attachment corners are visited by $R$ in the given order; also, the subpath, $\gamma_j$, of $R$ between $a_{j}$ and $a_{j+1}$ has $O(\sigma_G(a_j,a_{j+1}))$ vertices. For each component $C_i$, $a_i$ is the only vertex of $C_i$ visited by $R$. Moreover, the construction guarantees that $C_i$ lies to the right of $R$ when oriented from $a_1$ to $a_r$.

To prove that $R$ is planar, recall that in each round we extend $R$ by constructing  a path $\gamma_j$ that connects $a_j$ with $a_{j+1}$. We claim that at this point, no edge of $\gamma_j$ crosses the portion of $R$ constructed so far.
Indeed, because the value of $\varepsilon$ (\emph{resp.} $\delta$) increases (\emph{resp.} decreases) in each round, the edges of $\gamma_j$ that lie on the boundaries of some $\partial_\delta C_i$ or on $\partial_\varepsilon G^*$ cannot cross $R$ by the escape invariant.
Moreover, this invariant states that for each $a_i\notin R$, $R$ does not intersect~$\Delta_i$.
Because each cone $\Delta_i$ is narrowed in each round, the edges of $\gamma_j$ that lie on the boundary of this cone cannot cross $R$. Finally, because $\lambda < \delta$, the edges of $\gamma_j$ that lie on $\partial_\lambda C_j$ do not cross $R$. Therefore, we conclude that by concatenating $\gamma_j$ and $R$, we obtain a planar path.
\end{proof}

Figure~\ref{figure:big-example} illustrates the algorithm of Lemma~\ref{lemma:Path for connected augmentations} on a small example.  In this example, the path from $a_1$ to $a_2$ passes by $a_4$, so $R$ detours around $C_4$ in order to preserve the escape invariant at $a_4$.  After $R$ attaches to $a_2$ and $a_3$, it winds around components $C_2$ and $C_3$, respectively, in order to ensure that these components attach to the right of $R$.

\begin{figure}
   \centering{\includegraphics{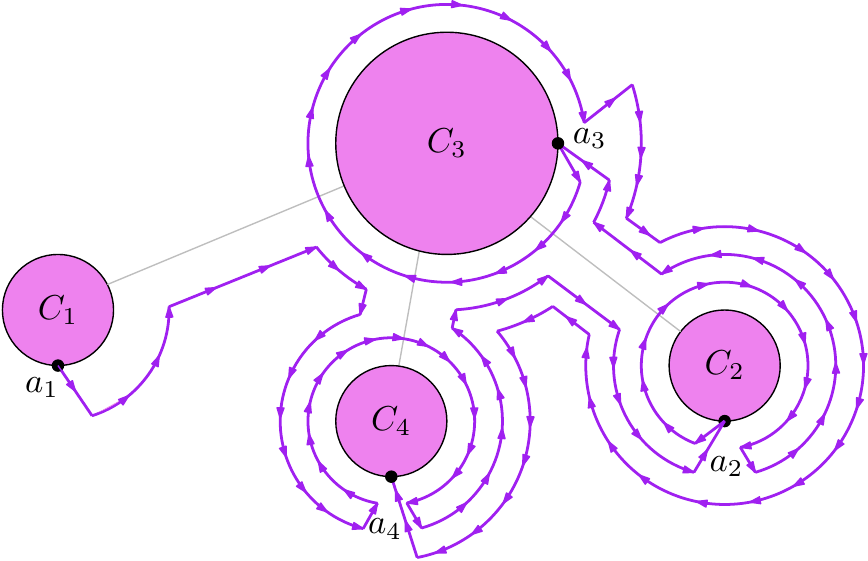}}
   \caption{An example of the algorithm for generating a spanning path that connects $a_1,\ldots,a_4$.}
   \label{figure:big-example}
\end{figure}

\subsection{Compatible drawings of planar graphs}
Let $\mathcal G$ be a planar graph with $n$ vertices and $r$ connected
components.  Let $G_1, \ldots, G_k$ be $k$ planar isomorphic drawings
of $\mathcal G$.  For now, we will assume that, in these drawings,
every component of $\mathcal G$ has at least one vertex incident to
the outer face.  
We show how to construct a compatible augmentation of $\mathcal G$
of size $O(nr^{1-1/k})$.

Let $\mathcal C_1, \ldots, \mathcal C_r$ be the connected components of $\mathcal G$.  Because $G_1,\ldots,G_k$ are isomorphic, we can select one attachment corner from each component in the drawing $G_1$, and this attachment corner also appears in each of $G_2,\ldots,G_k$. Thus, for each $j\in\{1,\ldots,r\}$, we choose an attachment corner $a_j$ of $\partial C_j$ such that $a_j$ is incident to the outer face of $C_j$.

For each $i\in \{1,\dots,k\}$, let $G_i^*$ be a connected augmentation of $G_i$, as defined in Section~\ref{section: connected augmentations}. For each $i\in\{1,\ldots,k\}$ and $j\in\{1,\ldots,r\}$, let $\rank_i(j) = \sigma_{G_i^*}(a_1, a_j)$. For each $j\in\{1,\ldots,r\}$, let $x_j\in \mathbb{R}^k$ be a point corresponding to the component $C_j$ such that $x_j = (\rank_1(a_j), \rank_2(a_j), \ldots, \rank_k(a_j))$. Let $X = \{x_1, \ldots, x_r\}\subset\R^k$ denote the resulting set of points. Lemma~\ref{lemma:Contained in integer grid} implies that $X$ is contained in an integer grid of side length $4n$.

Let $P$ be the shortest Hamiltonian path of $X$ under the $\ell_\infty$ norm. As before, because $X$ is contained in the $k$-dimensional integer grid of side-length $4n$ and $|X| = r$, the maximum ($\ell_\infty$) length of $P$ is $O(nr^{1-1/k})$. Note that the order of the points in $P$ induces an order of the components of $\mathcal G$ and hence an order of the attachment corners of each $G_i$.

\begin{theorem}\label{theorem:main}
For each $1\leq i\leq k$, we can construct a path $R_i$ of length $O(nr^{1-1/k})$ that connects every component of $G_i$ such that $G_i\cup R_i$ is planar. Moreover, for each $1\leq i<j\leq k$, $G_i\cup R_i$ is isomorphic to $G_j\cup R_j$.
\end{theorem}
\begin{proof}
By relabelling, let $(a_1, \ldots, a_r)$ denote the order of the attachment corners of $G_i$ induced by $P$.  Letting $d_j$ denote the $\ell_\infty$ distance between $x_j$ and $x_{j+1}$, we denote by $\mathcal{H}$ a path that passes through (the vertices corresponding to corners) $a_1, \ldots, a_r$ in this order, and that includes an additional $O(d_j)$ vertices between $a_{j}$ and $a_{j+1}$.  Thus, the number of vertices in $\mathcal{H}$ is proportional to the length of $P$, which is $O(nr^{1-1/k})$.

For each $G_i$, we use Lemma~\ref{lemma:Path for connected augmentations}
to draw $\mathcal{H}$ as a planar path, $R_i$, that connects $a_1,
\ldots, a_r$ in this order. By construction, 
\begin{align*}
  d_j & = \max\{|\rank_i(j+1) - \rank_i(j)|: i\in\{1,\ldots,k\}\} \\
    & \ge |\rank_i(j+1) - \rank_i(j)| \enspace ,
\end{align*}
so the $O(d_j)$ vertices in $\mathcal{H}$
between $a_j$ and $a_{j+1}$ are enough to draw the 
$O(\sigma_{G_i}(a_j,a_{j+1}))$ vertices in $R_i$ between $a_j$ and
$a_{j+1}$, since $\sigma_{G_i}(a_j,a_{j+1}) = \sigma_{G_i}(a_1,a_{j+1})-\sigma_{G_i}(a_1,a_{j}) = |\rank_i(j+1) - \rank_i(j)|$.



To conclude, each $R_i$ visits each component only at its attachment corner, the attachment corners of each $G_i$ are connected in the same order, and $R_i$ leaves every component to the right when oriented from $a_1$ to $a_r$. Therefore, $G_i\cup R_i$ is isomorphic to $G_j\cup R_j$ for each $1\leq i<j\leq k$.
\end{proof}

\subsection{Handling interior components}

In the preceding section, we assumed that the drawings $G_1,\ldots,G_k$ of $G$ were such that every component was incident to the outer face.
To see that the assumption is not necessary, observe that we can first use the preceding algorithm to connect all the components that do appear on the outer face using a polygonal path that is contained on the outer face.  The number of vertices used in this path is $O(n'r^{1-1/k})$, where $n'$ is the number of vertices on the outer face.

Next, for each interior face, $f$, that has multiple components
$C_1,\ldots,C_t$ on its boundary, we can use (a small modification of)
the preceding algorithm to connect $C_1,\ldots,C_t$ and the outer
boundary of $f$ using a path that is contained in $f$.  This path
has length $O(n_fr^{1-1/k})$.  We then repeat this step on each face.
The result is a connected augmentation of $\mathcal G$ whose total size
is $O(Nr^{1-1/k})$ where $N=O(n)$ is the total size of all faces.

\subsection{An algorithm}

We remark that Theorem~\ref{theorem:main} yields an efficient algorithm for constructing the augmentation $\mathcal{H}$, and even the drawings $H_1,\ldots,H_k$.  The main steps involved are: 

\begin{enumerate}
  \item Finding connected planar supergraphs $G_1^*,\ldots,G_k^*$
  of the drawings $G_1,\ldots,G_k$.  For each planar graph $G_i$, this
  can easily be done in $O(n\log n)$ time using, for example, a plane
  sweep algorithm that maintains the invariant that all components with
  a vertex to the left of the sweep-line are already joined by edges.
  Thus, this step takes $O(kn\log n)$ time.

  \item Constructing the point set $X$ and finding the path
  $P$.  Constructing $X$ takes $O(kn)$ time, while a path $P$ of length
  $O(n^{2-1/k})$ can be obtained from an (approximate) minimum spanning
  tree of $X$.  For constant values of $k$, an approximate MST can
  be computed in $O(n\log n)$ time using the algorithm of Calahan and
  Kosaraju \cite{callahan.kosaraju:faster}.  For larger values of $k$,
  the actual minimum spanning tree can be computed in $O(kn^2)$ time.

  \item Constructing each of the paths $R_1,\ldots,R_k$.
  Each of these paths is easily constructed in $O(n^{2-1/k})$ time
  once we have determined values of $\epsilon$, $\delta$, $\tau$, and
  $\lambda$ that are sufficiently small.  A more careful examination
  of our algorithm reveals that all that is really needed is
  a value of $\varepsilon$ such that $\partial_\epsilon G_i^*$ is
  simple, for each $i\in\{1,\ldots,k\}$.  Once we have this value of
  $\varepsilon$, the values of the remaining variables can taken from
  the set $\{i\varepsilon/3r:i\in\{1,\ldots,3r\}\}$.

  It turns out that a value $\varepsilon \le cm/n$, where $m$ is the
  minimum non-zero difference between $x$ coordinates or $y$ coordinates
  in $G_1,\ldots,G_k$, and $c$ is a constant, is sufficiently small.
  Thus, a suitable $\varepsilon$ can be computed in $O(kn\log n)$ time
  by sorting.
\end{enumerate}

This yields the following algorithmic result about connected augmentations:

\begin{theorem}
  An augmentation satisfying the conditions of Theorem~\ref{theorem:main}
  can be computed in $O(kn^2)$ time for any value of $k$.  If $k$ is
  constant, then the augmentation can be computed in $O(nr^{1-1/k})$
  time.
\end{theorem}

The latter result is worst-case optimal since, in the next section we will show that there exists inputs where every augmentation has size $\Omega(nr^{1-1/k})$.

\section{Lower Bounds}\label{section:Lower bound}

Our lower bounds are based on the following lemma. It says that we can find $k$ permutations of $\{1,\ldots,r\}$ such that for half the indices $i\in\{1,\ldots,r\}$, and every $j\in\{1,\ldots,r\}\setminus\{i\}$, there is a permutation in which $i$ and $j$ are at distance $\Omega(r^{1-1/k})$.

\begin{lemma}\label{lem:permutations}
Let $t=(1/2)^{1+1/k}\cdot(r-1)^{1-1/k}$.  There exists permutations $\pi^{(1)},\ldots,\pi^{(k)}$ of $\{1,\ldots,r\}$ such that for at least half the values of $i\in\{1,\ldots,r\}$ and for every $j\in\{1,\ldots,r\}\setminus\{i\}$,
\begin{equation}
 \max\left\{\left|\pi^{(s)}_i-\pi^{(s)}_j\right|\colon s\in\{1,\ldots,k\} \right\}
 \ge t \enspace .
     \label{eq:perm}
\end{equation}
\end{lemma}

\begin{proof}
  This proof is an application of the probabilistic method.  Select each
  of $\pi^{(1)},\ldots,\pi^{(k)}$ independently and uniformly from among
  all $r!$ permutations of $\{1,\ldots,r\}$.  Fix a particular index $i$
  and a particular index $j$.  For a particular $s\in\{1,\ldots,k\}$, the
  probability that $|\pi^{(s)}_i-\pi^{(s)}_j|\le t$ is at most $2t/(r-1)$
  since the set
  $\{\hat \jmath\in\{1,\ldots,r\} \colon |\pi^{(s)}_i-\pi^{(s)}_{\hat
   \jmath}|\le t\}$ is a random subset of at most $2t$ elements drawn without
  replacement from $\{1,\ldots,r\}\setminus \{i\}$.

  Therefore, since $\pi^{(1)},\ldots,\pi^{(k)}$ are chosen independently, 
  \[
    \Pr\left\{\max\left\{\left|\pi^{(s)}_i-\pi^{(s)}_j\right|\colon s\in\{1,\ldots,k\} \right\}\le t\right\} \le (2t/(r-1))^k = \frac{1}{2(r-1)} \enspace .
  \]
  In particular, the expected number of such
  $j\in\{1,\ldots,r\}\setminus\{i\}$ is at most $1/2$ so, by Markov's
  Inequality, the probability that there exists at least one such $j$
  is at most $1/2$.  Thus, with probability at least $1/2$, the index $i$
  satisfies \eqref{eq:perm} and therefore the expected number of indices $i\in\{1,\ldots,r\}$
  that satisfy \eqref{eq:perm} is $r/2$.  We conclude that there must exist
  some permutations $\pi^{(1)},\ldots,\pi^{(k)}$ that satisfy \eqref{eq:perm}
  for at least half the indices $i\in\{1,\ldots,r\}$.
\end{proof}

Using Lemma~\ref{lem:permutations}, we can prove a lower bound that matches the upper bound obtained in our general construction.

\begin{theorem}\label{thm:lower-bound}
  For every positive integer $n$ and every $r\in\{2,\ldots,\lfloor n/4\rfloor\}$,
  there exists a graph $\mathcal G$ having $n$ vertices, $r$ connected
  components, and $k$ isomorphic drawings $G_1,\ldots,G_k$ such that
  any compatible augmentation of $\mathcal G$ has size $\Omega(nr^{1-1/k})$.
\end{theorem}

\begin{proof}
Since the lemma only claims an asymptotic result, we may assume without
loss of generality that $r$ is even and that $2r$ divides $n$.

The graph $\mathcal G$ consists of $r$ disjoint paths,
$\mathcal{C}_1,\ldots,\mathcal{C}_r$, each of length $n/r$.  Each of the
drawings $G_1$,\ldots,$G_k$ has the vertices of $\mathcal G$ on the
same point and edge set. The point set consists of the vertices of
$r$ nested regular $n/r$-gons, $P_1,\ldots,P_r$, each centered at the
origin and having nearly the same size. Refer to Figure~\ref{figure:lower-bound} (left). More precisely, $P_1\subset
P_2\subset\cdots\subset P_r$ and the sizes are chosen so that any segment
joining two non-consecutive vertices of $P_i$ intersects the interior
of~$P_{i-1}$.
The drawings $G_1,\ldots,G_k$ are obtained from the permutations
$\pi^{(1)},\ldots,\pi^{(k)}$ given by Lemma~\ref{lem:permutations}.
In the drawing $G_x$, the path $\mathcal C_i$ is drawn on the vertices
of $P_{\pi^{(x)}_i}$. If $y=\pi^{(x)}_i$ is even, the drawing uses
all the edges of $P_y$ except the left-most edge.  If $y$ is odd, the
drawing uses all the edges of $P_y$ except the right-most edge.

\begin{figure}
  \centering{
    \begin{tabular}{c@{\hspace{1cm}}c}
      \includegraphics{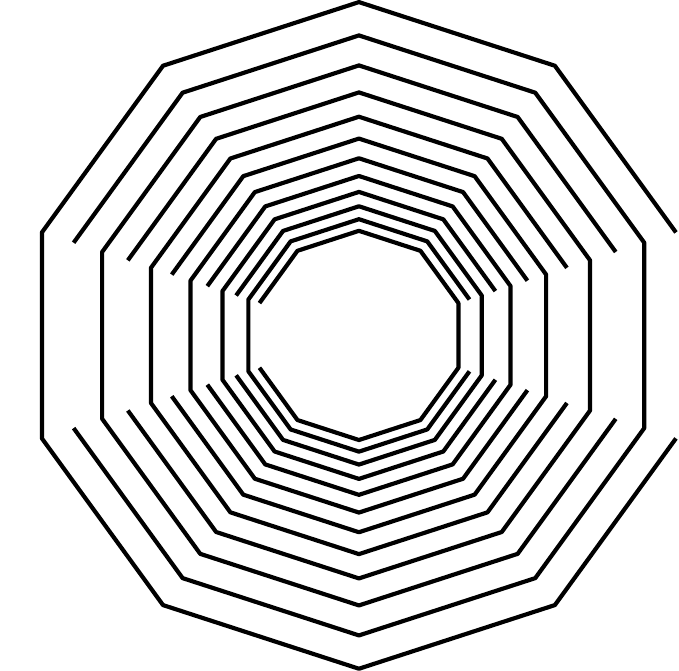} & \includegraphics{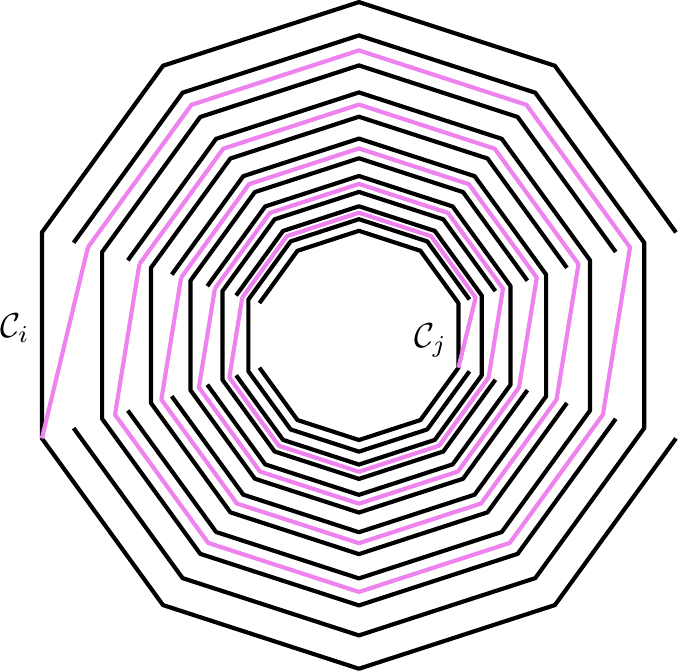}
    \end{tabular}
  }
  \caption{In the construction in Theorem~\ref{thm:lower-bound}, all drawings use the same set of vertices and line segments and the drawing of a path that joins $\mathcal C_i$ to $\mathcal C_j$ must travel around all paths drawn between the drawing of $\mathcal C_i$ and $\mathcal C_j$.}
  \label{figure:lower-bound}
\end{figure}

Now, without loss of generality, consider some edge-minimal compatible
augmentation $\mathcal H$ of $\mathcal G$.  For each component
$\mathcal{C}_i$ of $G$, let $T_i$ be any path in $\mathcal H$ that has
one endpoint on $\mathcal C_i$, one endpoint on some other component
$\mathcal{C}_j$, $j\neq i$, and no vertices of $\mathcal G$ in its
interior.
Now, for each of the $r/2$ indices $i\in\{1,\ldots,r\}$ that satisfy
\eqref{eq:perm}, the path $T_i$ joins a vertex of
$P_{\pi^{(s)}_i}$ to a vertex of $P_{\pi^{(s)}_j}$, $j\neq
i$, and $|\pi^{(s)}_i-\pi^{(s)}_i|\ge t$.  This path must
have length $\Omega(tn/r)$ since it has to ``go around'' the
paths between $P_{\pi^{(s)}_i}$ and $P_{\pi^{(s)}_j}$; see
Figure~\ref{figure:lower-bound} (right).

Thus far, we have shown that for at least $r/2$ values of
$i\in\{1,\ldots,r\}$, the component $C_i$ is the endpoint of a
path, $T_i$, of length at least $\Omega(tn/r)=\Omega(nr^{-1/k})$.
It is tempting to claim the result at this point, since
$(r/2)\cdot\Omega(nr^{-1/k})=\Omega(nr^{1-1/k})$. Unfortunately, there
is a little more work that needs to be done, since two such paths $T_i$
and $T_j$ may not be disjoint, so summing their lengths double-counts
the contribution of the shared portion.

To finish up we note that, since the augmentation $\mathcal{H}$ is minimal,
it is a tree; $\mathcal G$ contains no cycles, so any cycle in $\mathcal H$ contains an edge not in $\mathcal G$ that could be removed.  Now, observe that if we traverse the outer face of (any planar drawing of) $\mathcal H$ then we obtain a non-simple path, $P$, that traverses each edge of $\mathcal{H}$ exactly twice. If we consider the set of maximal subpaths of $P$ with no vertex of $\mathcal G$ in their interior, we obtain a set of $r$ edge-disjoint paths, $Q_1,\ldots,Q_{r}$ and, for every component $\mathcal C_i$ of $\mathcal G$, there is a vertex of $\mathcal C_i$ that is an endpoint of at least one such path.  Therefore, from the preceding discussion, the total length of $Q_1\ldots,Q_{r}$ is $\Omega(nr^{1-1/k})$.  But since each edge of $\mathcal H$ appears at most twice in these subpaths, we conclude that $\mathcal H$ has $\Omega(nr^{1-1/k})$ edges.  Since $\mathcal H$ is a tree, it has $\Omega(nr^{1-1/k})$ vertices.
\end{proof}

\section*{Acknowledgement}

This work was initiated at the \emph{Second Workshop on Geometry and Graphs},
held at the Bellairs Research Institute, March 9-14, 2014.  We are
grateful to the other workshop participants for providing a stimulating
research environment.


\bibliographystyle{plain}
\bibliography{augmentations}

\appendix
\newpage
\section{Shortest Tour in the Uniform Norm}
\label{app:uniform-norm}

Under the \emph{$\ell_\infty$ metric}, the distance between two points $x=(x_1,\ldots,x_k)$ and $y=(y_1,\ldots,y_k)$ is $\|x-y\|_\infty = \max\{|x_i-y_i|:i\in\{1,\ldots,k\}\}$.  

\begin{lemma}\label{lemma:tsp}
  Let $P$ be a set of $r\ge 2$ points contained in the $k$-dimensional
  cube $[0,1]^k$, for $k\ge 2$.  Then there exists a spanning path of $P$
  whose length under the $\ell_\infty$ metric is at most $cr^{1-1/k}$,
  where $c$ is a universal constant. In particular, $c$ does not depend
  on $k$ or $r$.
\end{lemma}

\begin{proof}
The following proof is a rehashing of an argument used by Moran
\cite{moran:on}. An argument of Few \cite{few:shortest}, which begins
by stabbing $[0,1]^k$ with a grid of $\Theta(r^{1-1/k})$ lines, could
also be used to establish the same asymptotic result.

First, we note that it is sufficient to upper-bound the $\ell_\infty$-length of the minimum spanning-tree of $P$, since this
can be transformed into a path of at most twice its length \cite{rosenkrantz.stearns.ea:analysis}.

For any point $p\in\R^k$, the \emph{uniform ball} of radius $d$ centered at $p$,
defined as
\[
    B_\infty(p,d)=\{q\in\R^k : \|p-q\|_\infty \le d\}
\]
is a cube of side-length $2d$ and has volume $(2d)^k$.  If $d< 1$ and
$p\in [0,1]^k$, then $B_\infty(p,d)\cap [0,1]^k$ contains a cube of side-length $d$, so $B_\infty(p,d)\cap [0,1]^k$ has volume at least $d^k$.  

The preceding implies that the set $P$ contains two points $p$ and $q$,
such that $\|p-q\|_\infty \le 2r^{-1/k}$; otherwise, one could pack $r$
disjoint cubes, each of volume greater than $1/r$ into $[0,1]^k$.

Now, we can construct a spanning tree of $P$ by repeatedly taking the
pair of points $p,q\in P$ that minimize $\|p-q\|_\infty$, adding the edge $pq$
to our spanning tree and then removing $q$ from $P$. Since, at the $i$th
step of this algorithm, the set $P$ contains $r-i+1$ points, the total $\ell_\infty$-length of all the edges added to this tree is
\begin{align*}
   \sum_{i=2}^{r} 2i^{-1/k} 
     & \le 2^{1-1/k} + \int_2^r 2x^{-1/k}\,\mathrm{d}x \\
     & = 2^{1-1/k} + (2/(1-1/k))\left(r^{1-1/k}-2^{1-1/k}\right) \\
     & \le 4r^{1-1/k} + O(1) \\
     & \le Cr^{1-1/k} \enspace ,
\end{align*}
for a sufficiently large constant $C$ and any $r,k\ge 2$. Thus, the result holds for $c=2C$.
\end{proof}

By uniformly scaling the point set $P$ by a factor of $n$, we obtain the following corollary of Lemma~\ref{lemma:tsp}, which is used in our algorithm:

\begin{corollary}\label{cor:tsp}
  Let $P$ be a set of $r\ge 2$ points contained in the cube $[0,n]^k$,
  for $k\ge 2$.  Then there exists a spanning path of
  $P$ whose length under $\ell_\infty$ metric is at
  most $cnr^{1-1/k}$, where $c$ is a universal constant. In particular,
  $c$ does not depend on $k$, $r$, or $n$.
\end{corollary}

\end{document}